\documentclass[11pt]{article}
\usepackage{amsmath,amsthm,amsfonts,amssymb,amsopn,amscd}     
\usepackage{epsfig}
\usepackage{graphicx,tikz}
\usepackage[all]{xy}
\usetikzlibrary{arrows}
\usepackage{setspace}
\usepackage{stmaryrd}

\usepackage{graphicx,color}
\usepackage{array}
\usepackage{tikz}
\usetikzlibrary{shapes.geometric}
\usetikzlibrary{decorations.pathmorphing}
\usetikzlibrary{arrows,shapes,positioning}
\usetikzlibrary{decorations.markings}
\tikzstyle arrowstyle=[scale=1]
\tikzstyle directed=[postaction={decorate,decoration={markings,mark=at position .65 with {\arrow[arrowstyle]{stealth}}}}]
\tikzstyle reverse directed=[postaction={decorate,decoration={markings,mark=at position .65 with {\arrowreversed[arrowstyle]{stealth};}}}]
\def\setminus{\smallsetminus}
\def\gA{{\mathfrak A}}
\def\gB{{\mathfrak B}}
\def\gC{{\mathfrak C}}
\def\gD{{\mathfrak D}}

\def\psl2r{{\rm PSL}(2,\RR)}
\def\f{{\varphi}}
\def\a{{\alpha}}
\def\b{{\beta}}
\def\l{{\lambda}}
\def\t{{\tau}}
\def\r{{\rho}}

\def\o{{\omega}}

\newtheorem{theorem}{Theorem}[section]
\newtheorem{lemma}[theorem]{Lemma}

\newtheorem{corollary}[theorem]{Corollary}

\newtheorem{proposition}[theorem]{Proposition}

\theoremstyle{definition} 

\newtheorem{remark}[theorem]{Remark}

\def\Diff{{\mathrm {Diff}}}

\def\RR{{\mathbb R}}

\def\A{{\mathcal A}}
\def\B{{\mathcal B}}
\def\C{{\mathcal C}} 

\def\D{{\mathcal D}}

\def\H{{\mathcal H}}
\def\I{{\mathcal I}}
\def\K{{\mathcal K}}
\def\M{{\mathcal M}}

\def\U{{\mathcal U}}

\def\L{\Lambda}

\newcommand{\be}{\begin{equation}} 
\newcommand{\ee}{\end{equation}}
\newcommand{\ba}{\begin{array}} 
\newcommand{\ea}{\end{array}}
\newcommand{\bea}{\begin{eqnarray}} 
\newcommand{\eea}{\end{eqnarray}}

\newcommand{\cref}[1]{Cor.~\ref{#1}}

\begin{document}

\title{\huge Non-Equilibrium Thermodynamics \\ and Conformal Field Theory}
 
\author{{\sc Stefan Hollands} \\
\small Institut f\"ur Theoretische Physik, Universit\"at Leipzig,
\\[-0.9mm]
\small Br\"uderstrasse 16, D-04103 Leipzig, Germany\\
[-0.9mm] \small {\tt stefan.hollands@uni-leipzig.de}\\
[2.2mm]{\sc Roberto Longo}\footnote{Supported in part by the ERC Advanced Grant 669240 QUEST ``Quantum Algebraic Structures and Models'', PRIN-MIUR, GNAMPA-INdAM and Alexander von Humboldt Foundation.}
 \\
\small Dipartimento di Matematica,
Universit\`a di Roma Tor Vergata,\\[-0.9mm]
\small Via della Ricerca Scientifica, 1, I-00133 Roma, Italy \\
[-0.9mm] \small {\tt longo@mat.uniroma2.it}
}

\maketitle

\vskip1cm

\begin{abstract}
We present a model independent, operator algebraic approach to non-equilibrium quantum thermodynamics within the framework of two-dimensional Conformal Field Theory. Two infinite reservoirs in equilibrium at their own temperatures and chemical potentials are put in contact through a defect line, possibly by inserting a probe. As time evolves, the composite system then approaches a non-equilibrium steady state that we describe. In particular, we re-obtain recent formulas of Bernard and Doyon~\cite{BD15}.
\end{abstract}

\newpage

\section{Introduction}
The purpose of non-equilibrium thermodynamics is to study physical systems that are not in thermodynamic equilibrium but can be basically described by thermal equilibrium variables. It thus deals with systems that are in some sense near equilibrium. 
Although the research on non-equilibrium thermodynamics has been effectively pursued for decades with important achievements, the general theory still missing. 
The framework is even more incomplete in the quantum case, non-equilibrium quantum statistical mechanics. 

Non-equilibrium thermodynamics deals with inhomogeneous systems. A typical model system is given by two infinite reservoirs, initially in equilibrium at different temperatures and different chemical potentials, set in contact at the boundary with an energy flux from one reservoir to the other; possibly the global system may incorporate a probe between the two reservoirs.

The purpose of this paper is to provide a general, model independent scheme for the above situation in the context of quantum, two dimensional Conformal Quantum Field Theory. As we shall see, we provide the general picture for the evolution towards a non-equilibrium steady state, and obtain in particular formulas derived in \cite{BD15}.

We use the Operator Algebraic description of Conformal Field Theory (see \cite{KL04}), in particular the phase boundary description of~\cite{BKLR,BKLR16,BKL}, and the study of the thermodynamical equilibrium states in \cite{CLTW1,CLTW2}. In this way we get a transparent picture of the system and its states as time evolves. Our setup is described schematically by the following spacetime-diagram fig.~\ref{DOC1}. 

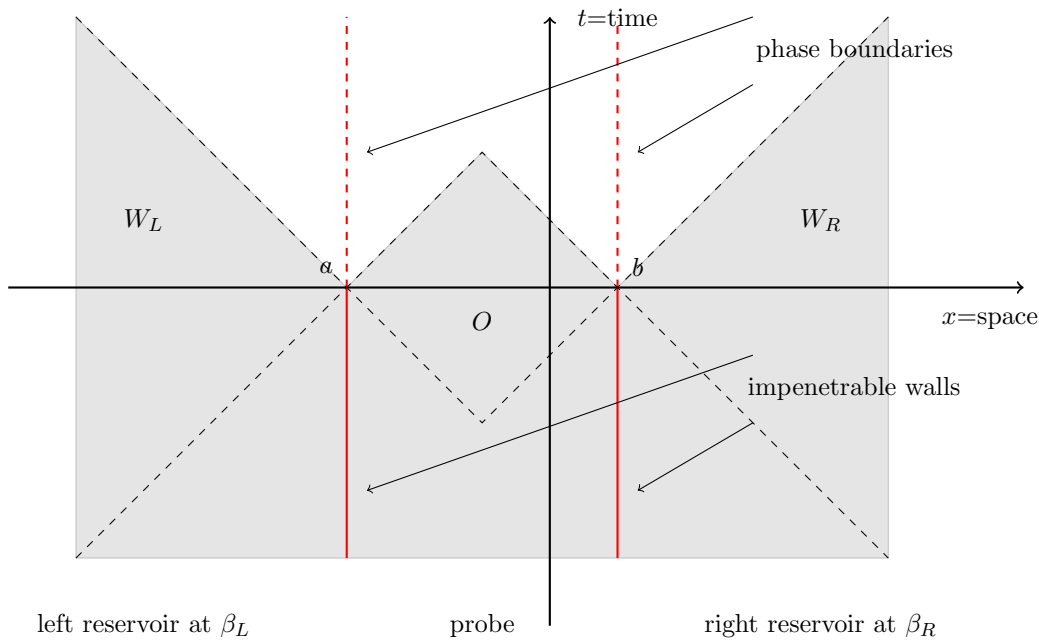
\begin{figure}[h] 
\begin{center}
\begin{tikzpicture}[scale=.9, transform shape]
\filldraw[fill=gray,opacity=.2,draw=black] (-7,0) -- (5,0) -- (5,8) -- (1,4) -- (-1,6) -- (-3,4)--(-7,8)--(-7,0);
\draw[dashed, draw=black] (-7,8) -- (-3,4) -- (-1,6) -- (1,4) -- (5,8);
\draw[dashed, draw=black] (-7,0) -- (-3,4) -- (-1,2) -- (1,4) -- (5,0);
\draw[red, thick, dashed] (1,4) -- (1,8);
\draw[red, thick, dashed] (-3,4) -- (-3,8);
\draw[red,thick] (-3,0) -- (-3,4);
\draw[red,thick] (1,0) -- (1,4);
\draw[->,black, thick] (0,-1) -- (0,8);
\draw[->,black, thick] (-8,4) -- (7,4);
\draw[->,black] (3,8) -- (-2.7,6);
\draw[->,black] (3,7) -- (1.3,6);
\draw[->,black] (3,3) -- (-2.7,1);
\draw[->,black] (3,2) -- (1.3,1);
\node at (4.5,7.5) {phase boundaries};
\node at (4.5,2.5) {impenetrable walls};
\node at (6.5,3.5) {$x$=space};
\node at (1,8) {$t$=time};
\node at (1.3,4.3) {$b$};
\node at (-3.3,4.3) {$a$};
\node at (4,5) {$W_R$};
\node at (-6,5) {$W_L$};
\node at (-1,3.5) {$O$};
\node at (-1,-1) {probe};
\node at (4,-1) {right reservoir at $\beta_R$};
\node at (-6,-1) {left reservoir at $\beta_L$};
\end{tikzpicture}
\end{center}
\caption{
\label{DOC1}
Spacetime diagram of our setup. The initial state $\psi$ is set up in the shaded region before the system is in causal contact with the phase boundaries. 
In the shaded regions to the left/right of the probe, we have a thermal equilibrium state at inverse temperatures $\beta_L/\beta_R$. In the diamond shaped 
shaded region $O$, we have an essentially 
arbitrary probe state. 
}
\end{figure}

Before $t=0$, we have a stationary probe situated in the $x$-interval $(a,b)$. In the interval to the left  $(-\infty,a)$, we have a reservoir characterized by a thermal equilibrium state at inverse temperature $\beta_L=1/T_L$, whereas in the interval $(b, \infty)$ to the right  we have a reservoir at another inverse temperature $\beta_R=1/T_R$. The probe is characterized by yet another, essentially arbitrary, state. These three subsystems are perfectly shielded from each other by impenetrable walls up to time $t=0$. At time $t=0$, we now replace the walls by transmissive phase boundaries. Since the propagation speed is finite in a relativistic quantum field theory, the states to the left and right 
continue to be described by equilibrium states inside the left and right wedges, $W_{L/R}$, and by the probe state inside the diamond $O$, i.e. outside the causal domains of dependence of the phase boundaries. The question is then how the state of the system is described towards the future of these regions, where the transparent nature of the boundaries after $t=0$ can be seen.

Our result is that this state approaches a non-equilibrium steady state, $\omega$, in the sense of \cite{Ru00} (see the main text for the precise definition of this notion): 
If $\psi$ is the initial state of the system (with impenetrable walls), $Z$ is any local observable of the system, and $\tau_t$ the time-translation automorphism, then $\lim_{t \to \infty} \psi(\tau_t(Z)) = \omega(Z)$.  The non-equilibrium steady state $\omega$ is determined uniquely by the temperatures characterizing the reservoirs (and the nature of the 
transmissive phase boundaries), and thus does {\em not} depend on 
the arbitrary state of the probe. In particular, the expectation value of the momentum density operator $\psi(T_{tx}(t,x))$ as $t \to \infty$ approaches 
the expectation value in the non-equilibrium steady state $\omega$, which in turn is given by a simple formula 
\be\label{asyT0}
\lim_{t\to+\infty} \psi(T_{tx}(t,x)) = \frac\pi{12}(c_L^{} \beta_L^{-2} - c_R^{} \beta_R^{-2}) \ ,
\ee
where $c_{L,R}$ are the central charges of the left/right moving sub-sectors of the theory. 

We also derive a similar result in case the reservoirs are characterized not just by a temperature, but also by a chemical potential. This requires the 
underlying conformal field theory to contain a current. Here we consider for simplicity a $U(1)$-current (see~\cite{BMT}), but the result could presumably be generalized to 
general current algebras corresponding to a some simple compact Lie algebra such as $\frak{su}(N)_k$. 

A setup similar to ours has previously been studied by \cite{BD15} (see also \cite{BD12,BD16}), where, in particular, asymptotic flux formulas such as \eqref{asyT0} were also obtained. These authors do not use the
mathematical formalism of operator algebras and conformal nets as we do, but instead work in the setup of vertex operator algebras. Furthermore, rather than working directly in the thermodynamic limit, they prefer to look at limits of finite systems described by density matrix states and ``scattering matrices'' of the underlying CFT.
The precise mathematical status of these constructions does not seem to be completely settled~\footnote{See e.g.
remarks 4.1 and 4.2 of \cite{BD15}.}, and, furthermore, the scope of our analysis seems to be broader in some respects, for instance in that we can allow for probe between the reservoirs, and probably also in the type of defects that we can handle. Nevertheless, the overall physical picture that emerges in their approach seems broadly compatible with ours.

The plan of this paper is as follows: We first 
provide a very concise summary on some background such as the ``geometric'' thermal equilibrium states, conformal nets, and phase boundaries in Sect.~2. Then we treat the situation with two reservoirs, one phase boundary, but with no probe nor chemical potential in Sect.~3.1.1. The latter is included after a brief summary of relevant results on the $U(1)$-current algebra in Sect.~3.1.2. The case with a probe is considered in Section~3.2. We end the paper with Sect.~4, where we demonstrate the relativistic KMS-condition for the geometric KMS-states, a technical result needed in the proofs in Sect.~3.

\section{Preliminaries}
We begin to recall some essential facts upon which our analysis will be based.

\subsection{Non-equilibrium  steady states, see \cite{Ru00}}
\label{Sect1}
As is well known, the thermal equilibrium states in Quantum Statistical Mechanics at infinite volume
are the KMS states \cite{HHW}.

Let $\gA$ be a $C^*$-algebra, $\tau=\{\tau_t\}_{t \in \mathbb{R}}$ a one-parameter group of automorphisms of $\gA$ and $\gA_0$ a dense $^*$-subalgebra of $\gA$.  A
state $\omega$ of $\gA$ is a positive, linear functional $\omega: \gA \to \mathbb{C}$ that is normalized, $\omega(1) = 1$.  A state $\omega$ is called ``KMS state for $\tau$ at inverse temperature $\b>0$'' if, for any $X,Y\in\gA_0$, there is a function $F_{XY}\in A(S_\b)$ such that 
\begin{itemize}
\item[$(a)$] $F_{XY}(t)= \omega\big(X\tau_t(Y)\big)$,
\item[$(b)$] $F_{XY}(t+i\b)=\omega\big(\tau_t(Y)X\big)$,
\end{itemize}
where $A(S_\b)$ is the algebra of functions analytic in the strip
$S_\b=\{0<\Im z <\b\}$, bounded and continuous on the closure $\bar
S_\b$.  Properties $(a)$
and $(b)$ then actually hold for all $X,Y\in\gA$. For a finite system, the automorphisms are 
implemented by a trace-class Hamiltonian $H$, i.e. $\tau_t(a) = e^{itH} a e^{-itH}$. Then 
the density matrix $\rho = e^{-\beta H}/Z$ defines a KMS state via
$\omega(a) = {\rm tr}(a \rho)$.  The notion of KMS-state generalizes the usual notion of Gibbs-state 
to infinite systems~\cite{HHW}, where $H$ is not of trace-class. 

Let us now consider a non-equilibrium statistical mechanics situation.
Suppose a quantum system $\Sigma$ is interacting with a set
of infinite reservoirs $R_{k}$ that are in equilibrium at different
temperatures $\b^{-1}_k$.  In this context, a natural class of stationary
non-equilibrium states occur, the {\it non-equilibrium steady states}, see 
\cite{Ru00}.  If we denote as above the observable $C^*$-algebra by $\gA$ and the time
evolution automorphism group by $\t$, by definition a non-equilibrium steady state
$\omega$ of $\gA$ satisfies property $(a)$ in the KMS condition, namely $F_{XY}(t)=\omega\big(X\tau_t(Y)\big)$ extends
holomorphically on a strip $S_\b$ (and continuosly on $0\leq \Im z < \b$) for any
$X,Y$ in a dense $^*$-subalgebra $\gA_0$ of $\gA$,
but property $(b)$ does not necessarily hold. Here $\b=\text{min}\b_k$.

A first example of non-KMS, non-equilibrium steady state
is provided by the tensor product of KMS
states at different temperatures; in this case the parameter $\b$ is
clearly the minimum of the inverse temperatures.  

\subsection{Two-dimensional conformal nets}
\label{Sect2}
We recall here some basic definitions and properties of a conformal net on the Minkowski plane, see \cite{CLTW1}.

Let $M$ be the two-dimensional Minkowski spacetime. 
A \emph{double cone} $O$ is a non-empty open subset of
$M$ of the form $O=I_+\times I_-$ with $I_{\pm}$
bounded open intervals of the light ray lines $\L_\pm\equiv\{ (t,x): t\pm x = 0\}$; we also set $u\equiv t + x$, $v\equiv t - x$ and denote by $\K$ the set of double cones of $M$.

The M\"{o}bius group $\psl2r$ acts on each compactified light ray line $\L_\pm\cup\{\infty\}$ by linear fractional transformations, hence we  have a local (product) action of
$\psl2r\times \psl2r$ on $M$.

A \emph{local M\"{o}bius covariant net} $\B$ on $M$ is a map
\[
\B:O\in\K\mapsto\B(O)
\]
where the $\B(O)$'s are von Neumann algebras on a fixed Hilbert space $\H$, with the following properties:
\begin{itemize}

\item \emph{Isotony.} $O_1\subset O_2\implies
\B(O_1)\subset\B(O_2)$.

\item \emph{M\"{o}bius covariance.} There exists a unitary representation
$U$ of $\overline{\psl2r}\times\overline{\psl2r}$
on $\H$ such that, for every double cone $O\in\K$,
\[
U(g)\B(O)U(g)^{-1} = \B(gO),\quad g\in\U,
\]
with $\U\subset\overline{\psl2r}\times\overline{\psl2r}$
any connected neighborhood of the identity
such that $gO\subset M$ for all $g\in\U$. Here $\overline{\psl2r}$
denotes the universal cover of $\psl2r$.

\item \emph{Vacuum vector.} There exists a unit $U$-invariant vector
$\Omega$, cyclic for $\bigcup_{O\in\K}\B(O)$.

\item \emph{Positive energy.} The one-parameter unitary subgroup of $\{U(t)\}_{t \in \mathbb{R}}$
corresponding to time translations ${\bf x} \mapsto {\bf x} + t{\bf e}$ with ${\bf e} \in M$ any future pointing timelike vector, 
has positive generator.

\end{itemize}

$\bullet$ $\B$ is \emph{local} if $\B(O_1)$ and $\B(O_2)$ commute element-wise if $O_1$ and $O_2$ are spacelike separated.
\smallskip

\noindent
A \emph{local conformal net} $\B$ on $M$ is a local M\"{o}bius covariant
net $\B$ such that the unitary representation $U$ extends to a
projective unitary representation of the group of global, orientation preserving conformal diffeomorphisms of the Einstein cylinder (a time cover of the 2-torus compactification of $M$, see \cite{BGL}). In particular
\be\label{diff}
U(g)\B(O)U(g)^{-1} = \B(gO)\ ,\quad g\in\Diff(\RR)\times\Diff(\RR)\ ,
\ee
if $O\in\K$. We further assume that
\begin{equation}\label{loc}
U(g)XU(g)^{-1} = X\ ,\quad g\in\Diff(\RR)\times\Diff(\RR)\ ,
\end{equation}
if $X\in\B(O)$, $g\in\Diff(\RR)\times\Diff(\RR)$ and $g$ acts
identically on $O$. 

Given a local M\"obius covariant net $\B$ on $M$ and a bounded interval
$I$ of the chiral line $\L_+$ we set
\begin{equation}\label{movers}
\A_{+}(I)\equiv \bigcap_{O=I\times J}\B(O)
\end{equation}
(intersection over all intervals $J$ of the chiral line $\L_-$), and we analogously
define $\A_-$. By identifying the light-ray lines $\L_{\pm}$ with $\RR$ we then get two
local nets $\A_{\pm}$ on $\RR$, \emph{the chiral components of $\B$}.
They extend to local M\"obius covariant nets on $S^1$. The Hilbert space
$\H_{\pm}\equiv\overline{\A_{\pm}(I)\Omega}$
is independent of the interval $I\subset\L_{\pm}$ and $\A_{\pm}$
restricts to a (cyclic) M\"obius covariant conformal net on $\H_{\pm}$. 
Moreover $\A_\pm$ contains the diffeomorphism symmetries, i.e. the Virasoro subnet.

Let us assume $\B$ to  be a local conformal net on $M$.
Set 
\[
\A(O)\equiv \A_+(I_+)\vee\A_-(I_-)\ , \quad O=I_+\times I_- \ , 
\]
with $\A_\pm$ given by \eqref{movers} (or, more generally, a subnet of $\A_\pm$ containing the Virasoro subnet).
Then $\A$ is a local conformal, irreducible subnet of
$\B$. Both $\A(O)$ and $\B(O)$ are factors for any $O\in\K$, and
there exists a consistent family of vacuum preserving
conditional expectations $\varepsilon_{O}:\B(O)\to\A(O)$ and $\A_+(I_+)\vee\A_-(I_-)$ is naturally isomorphic to the tensor product
$\A_+(I_+)\otimes\A_-(I_-)$, $O= I_+\times I_-$.

We shall say that $\B$ is \emph{completely rational}
if the two associated chiral nets $\A_\pm$ in \eqref{movers} are completely rational \cite{KLM}.
If $\B$ is completely rational, the following facts hold in particular \cite{KLM,KL04}:
\begin{itemize}
\item The inclusion of factors $\A(O)\subset\B(O)$ has finite Jones index, $O\in\K$;
\item Both $\A$ and $\B$ have finitely many irreducible sectors and all of them have finite dimension. Every irreducible sector $\r$ of $\A$ has the form $\r =\r_+\otimes\r_-$ with $\r_\pm$ an irreducible sector of $\A_\pm$.
\end{itemize}
Let $E\subset M$ be an open region of the spacetime $M$. We denote by $\K(E)$ the set of double cones $O$ of $M$ with closure $\bar O\subset E$. Thus $\K(M)=\K$ in particular.

With $\A$ a net of von Neumann algebras on $M$, we shall consider the $C^*$-algebra $\gA(E)$ generated by the von Neumann algebras $\A(O)$ with $O\in\K(E)$, and also set $\gA\equiv\gA(M)$. Similarly, $\gB(E)$ denotes the $C^*$-algebra associated with $E$ by the net $\B$, and so on (in other words we use a calligraphic letter to denote a net of von Neumann algebras and the corresponding gothic letter for the $C^*$-algebra associated by the net to a region). 

\subsubsection{KMS states in CFT, see \cite{CLTW1,CLTW2}}

We now recall the definition of the {\it geometric KMS state} \cite{CLTW1}. Let $\C$ be a local conformal net on $\mathbb R$. The exponential map gives an isomorphism of $\gC(\mathbb R)$ with $\gC(\mathbb R_+)$.
The vacuum state $\omega$ restricts to a KMS state $\omega|_{\gC(\mathbb R_+)}$ w.r.t. dilations at inverse temperature $2\pi$ (this statement is a one-dimensional analogue 
of the Bisognano-Wichmann theorem~\cite{BW}). 
The geometric KMS state (at inverse temperature $2\pi$) on $\C$ is the state $\f$ on $\gC(\mathbb R)$ obtained as the pullback of $\omega |_{\gC(\mathbb R_+)}$ by the exponential map. Clearly $\f$ is a KMS state on $\gC(\mathbb R)$ w.r.t. translations at inverse temperature $2\pi$. 

We further obtain a  KMS state w.r.t. $\tau$ at any given inverse temperature $\b > 0$ by $\f_\b = \f\cdot \delta_\lambda$, where $\delta_\lambda$ is the dilation automorphism of $\gC(\mathbb R)$ by $\lambda = \b/2\pi$; we call this state the ``geometric KMS state of $\C$ at inverse temperature $\b$.'' For convenience, in the following the geometric KMS state means the geometric KMS state at inverse temperature $\b = 1$, unless we specify a different temperature.

If $\B$ is a conformal net on $M$, the geometric KMS state w.r.t. time translations is similarly constructed. If $\A = \A_+\otimes \A_-$ is the subnet generated by the chiral subnets we have $\f = \f^+\otimes\f^-\cdot\varepsilon$, where $\f^\pm$ is the geometric KMS state on $\A_\pm$ and $\varepsilon : \gB\to \gA$ is the natural conditional expectation (i.e. $\varepsilon|_{\B(O)} = \varepsilon_O$).

A basic result for KMS states with respect to translations is the following.
\begin{theorem}{\rm \cite{CLTW1}.}\label{gstate}
If $\B$ is a completely rational local conformal net on $M$, there exists a unique $\b$-KMS state $\f_\b$ of $\mathfrak B$ w.r.t. the time translation group $\tau$ at any given inverse temperature $\b>0$. The state $\f_\b$ is the geometric $\b$-KMS state of $\gB$.
\end{theorem}
By construction, the geometric KMS state is locally normal. Notice however that, for any conformal net with the split property, every KMS state w.r.t. translations is in fact locally normal (see \cite{CLTW1,CLTW2} and references therein).

The above theorem holds true, with the same proof, also if $\B$ is non local, but relatively local w.r.t. the completely rational, chiral subnet
 $\A=\A_+\otimes\A_-$.
In general, the following holds (see \cite{CLTW1}).
\begin{proposition}\label{genKMSstate}
Let $\B\supset \A=\A_+\otimes\A_-$ be a local conformal net on $M$ and $\varphi$ an extremal $\b$-KMS state of $\mathfrak B$ w.r.t. the time translation group $\tau$. Then 
$\varphi$ is locally normal  and $\f|_\gA = \f^+\otimes\f^-$, where $\f^\pm$ is an extremal $\b$-KMS state on $\gA_\pm$.
\end{proposition}
Also Proposition \ref{genKMSstate} holds true if the irreducible extension of $\B$ of $\A$ is non-local but relatively local with respect to $\A$.
\subsubsection{Chemical potential, see \cite{AHKT,L01}}
\label{ChemPot}
Let $\A$ be a local conformal net on $M$ (or on $\mathbb R$) and $\f$ an extremal $\b$-KMS state on $\gA$ w.r.t. the time translation group $\t$. The von Neumann algebra $\M\equiv\pi_\f(\gA)''$ in the GNS representation $\pi_{\f}$ is a factor.

Let $\r$ be an irreducible DHR localized endomorphism of $\A$ with finite index, namely the dimension $d(\rho)$ of $\rho$ is finite \cite{L89}. We assume that $\r$ is normal, namely it extends to a normal endomorphism of $\M$; this automatically holds under general assumptions, for example if $\f$ satisfies essential duality, i.e. $\pi_\f\big(\gA(W_{L/R})\big)'\cap\M = \pi_\f\big((\gA(W_{R/L})\big)''$.

Let $U$ be the time translation unitary covariance cocycle in $\gA$ for the endomorphism $\r$ defined by
\[
{\rm Ad}U(t)\cdot\t_t \cdot \r  = \r \cdot \t_t \ , \quad t\in \mathbb R \ , 
\]
with $U(t+s) = U(t) \tau_t\big(U(s)\big)$ (cocycle relation). 
The choice of the $U$ is unique up to a phase and unique if we assume, as we will do from now, that $U$ is the restriction of the M\"obius covariance unitary cocycle (see \cite{L01}).

Denote by $\Phi_\r$ the left inverse of $\r$ on $\M$; then $U$ is equal up to a phase to a Connes Radon-Nikodym cocycle \cite{Con1}, namely there exists $\mu_\r(\f)\in\mathbb R$ such that
\be\label{chempot}
U(t) = e^{-i2\pi\mu_\r(\f)t}d(\r)^{-i\b^{-1}t}\big(D\f\cdot\Phi_\r : D\f\big)_{- \b^{-1}t} \ .
\ee
$\mu_\r(\f)$ is the {\it chemical potential} of $\f$ w.r.t. the charge $\r$.

The geometric $\b$-KMS state $\f_0$ has zero chemical potential \cite{L97}. 
By the holomorphic property of the Connes Radon-Nikodym cocycle \cite{Con1}, we then have \cite{L01}:
\be\label{fc}
e^{2\pi\b\mu_\r(\f)} =  {\underset{t\, \longrightarrow\,  i\b}{\rm anal.\, cont.\,}\f\big(U(t)\big)}\big/
{\underset{t\, \longrightarrow\,  i\b}{\rm anal.\, cont.\,}\f_0\big(U(t)\big)} \ ,
\ee
which in fact holds for any choice of the phase for the unitary time covariance cocycle $U$.
\subsection{Phase boundaries, see \cite{BKLR16}}
\label{phaseb}

Let $M_L \equiv \{(t,x): x <0 \}$ and $M_R \equiv \{(t,x): x > 0\}$ be the left and the right half Minkowski plane.

A  (transmissive) {\it phase boundary} is given by specifying two local conformal nets
$\B^L$ and $\B^R$ on $M$, covariantly represented on the same Hilbert space $\H$;
$\B^L$ and $\B^R$ both contain a common chiral subnet $\A=\A_+\otimes\A_-$.
Initially $\B^{L/R}$ is defined on $M_{L/R}$
\[
\K(M_L)\ni O\mapsto \B^L(O)\ ; \qquad \K(M_R)\ni O \mapsto \B^R(O)\ ,
\]
yet $\B^{L/R}$ extends on the entire $M$ by covariance.
Indeed, the chiral nets $\A_\pm$ on $\mathbb R$ contain
the unitaries implementing the local diffeomorphisms, and hence both nets $\B^L$ and $\B^R$ share the same unitary representation of the symmetry group $\Diff(\mathbb R)\times\Diff(\mathbb R)$.

Causality requires that the algebras  
$\B^L(O_1)$ and $\B^R(O_2)$ commute whenever $O_1\subset M_L$
and $O_2\subset M_R$ are spacelike separated. By diffeomorphism 
covariance, $\B^R$ is thus right local with respect to $\B^L$, i.e. if $O_1$ is spacelike to $O_2$ {\em and}
$O_2$ is to the left of $O_R$, then we have $[\B^L(O_2), \B^R(O_1)] = 0$.  

Given a phase boundary, we consider the von Neumann algebras generated by $\B^L(O)$ and $\B^R(O)$:
\be
\label{eq:add}
\D(O)\equiv\B^L(O)\vee\B^R(O)\ , \quad O\in\K\ .
\ee
$\D$ is another extension of $\A$, but $\D$ is in general
non-local, but relatively local w.r.t. $\A$. $\D(O)$ may have non-trivial center. In the completely rational case, $\A(O)\subset \D(O)$ has finite Jones index, so the center of $\D(O)$ is finite dimensional; by standard arguments, we may cut down the center to $\mathbb{C}$ by a minimal projection of the center, and we may then assume $\D(O)$ to be a factor, as we will do for simplicity in the following.

\section{Non-equilibrium states in CFT}

We now consider a non-equilibrium quantum thermodynamical system described within conformal field theory and discuss the approach to a non-equilibrium steady state as the system evolves in time.
\subsection{Case with no chemical potential and no probe}\label{main}

Let us consider two local conformal nets $\B^L$ and $\B^R$ on the Minkowski plane $M$ and both containing the same chiral net 
$\A=\A_+\otimes\A_-$. In this section $\B^{L/R}$ is completely rational, and we use the uniqueness of the geometric KMS state (Thm. \ref{gstate}), later we get on the case where chemical potentials are present.
\medskip

\noindent
{\bf Before contact.} We assume that the two systems $\B^L$ and $\B^R$ are, separately, each in a thermal equilibrium state, possibly at different temperatures. Namely we consider the KMS state $\f^{L/R}_{\b_{L/R}}$ on $\gB^{L/R}$ at inverse temperature $\b_{L/R}$ with respect to the translation automorphism group $\tau$, possibly with $\b_L \neq \b_R$. 

At the moment, the two systems $\B^L$ and $\B^R$ live independently in their own half plane $M_L$ and $M_R$ and their own Hilbert space.
The composite system is described by the net on $M_L\cup M_R$ given by
\be\label{B}
\K(M_L)\ni O\mapsto \B^L(O) \, , \qquad   \K(M_R)\ni  O\mapsto \B^R(O) \, .
\ee
The $C^*$-algebra of the composite system is 
$
\gB^L(M_L)\otimes\gB^R(M_R)
$
and the state of the system is
\[
\f = \f^L_{\b_L}|_{\gB^{L}(M_{L})} \otimes \f^R_{\b_R}|_{\gB^{R}(M_{R})} \ ;
\]
$\f$ is a stationary state, a non equilibrium steady state, but not a KMS state.

We will denote by $V_\pm=\{(t,x) : \pm t > |x|\}$ the forward/past {\it light cone} and by $W_{L/R}=\{(t,x): \mp x >|t|\}$ the left/right {\it wedge} in the two-dimensional Minkowski space $M$.
\medskip

\noindent
{\bf After contact.} 
At time $t=0$ we put the two systems $\B^L$ on $M_L$ and $\B^R$ on $M_R$ in contact through a totally transmissible phase boundary and the time-axis the defect line. We are in the framework of Sect. \ref{phaseb}, with $\B^L$ and $\B^R$ now nets on $M$ acting on a common Hilbert space $\H$. With  $O_1\subset M_L$, $O_2\subset M_R$ double cones, the von Neumann algebras $\B^L(O_1)$ and $\B^R(O_2)$ commute if $O_1$ and $O_2$ are spacelike separated, so $\gB^L(W_L)$ and $\gB^R(W_R)$ commute. 

We want to describe the state $\psi$ of the global system after time $t=0$. As above, we set
\[
\D(O) \equiv \B^L(O)\vee\B^R(O)\ , \quad O\in\K\ .
\]
The origin ${\bf 0}$ is the only $t=0$ point of the defect line; the observables localized in the causal complement $W_L\cup W_R$ of the $\bf 0$ thus do not feel the effect of the contact, so $\psi$ should be a natural state on $\gD$ that satisfies
\be\label{psiW}
\psi |_{\gB^L(W_L)} = \f^L_{\b_L} |_{\gB^L(W_L)}, \quad \psi |_{\gB^R(W_R)} = \f^R_{\b_R} |_{\gB^R(W_R)}\ .
\ee
In G, $\psi$ is to be a {\it local thermal equilibrium state} on $W_{L/R}$ in the sense of \cite{BOR}.

Since $\gB^L(M_L)$ and $\gB^R(M_R)$ are not independent, the existence of such state $\psi$ is not obvious. Clearly the $C^*$-algebra on $\H$ generated by $\gB^L(W_L)$ and $\gB^R(W_R)$ is naturally isomorphic to $\gB^L(W_L)\otimes\gB^R(W_R)$ 
($\gB^L(W_L)''$ and $\gB^R(W_R)''$ are commuting factors) and the restriction of  
$\psi$ to it is the product state $\f^L_{\b_L}|_{\gB^{L}(W_L)} \otimes \f^R_{\b_R}|_{\gB^{R}(W_R)}$.
\begin{lemma}\label{autom}
Let $\C$ be a (possibly non local) conformal net on $\mathbb R$. Given $\lambda_- , \lambda_+ >0$, there exists an automorphism $\a\equiv \a_{\lambda_-,\lambda_+}$ of the $C^*$-algebra $\gC(\mathbb R\setminus\{0\})$ such that 
\[
\a|_{\gC(-\infty,0)} = \delta_{\lambda_-}\ , \quad \a|_{\gC(0, \infty)} = \delta_{\lambda_+}\ ,
\]
where $\delta_\lambda$ is the $\lambda$-dilation automorphism of $\gC(\mathbb R)$.
\end{lemma}
\begin{proof}
Let $h:\mathbb R\to\mathbb R$ be the function $h(a) = \lambda_- a$ if $a\leq 0$ and $h(a) = \lambda_+ a$ if $a\geq 0$.

With $I\subset \mathbb R$ a bounded interval such that $0\notin \bar I$, choose any diffeomorphism $\tilde h$ of $S^1 = \mathbb R\cup\{\infty\}$ such that $\tilde h$ is equal to $h$ on $I$. With $U$ the unitary representation of $\Diff(S^1)$ associated with $\C$, the map Ad$U(\tilde h)$ is an isomorphism $\a_I$ of $\C(I)$ onto $\C(\tilde hI)= \C(hI)$ that does not depend on the choice of $\tilde h$. We can then define the automorphism $\a$ on $\bigcup_I \C(I)$ ($0\notin\bar I$) by $\a_I |{_{\C(I)}} = \a_I$,  hence $\a$ is defined on $\gC(\mathbb R \setminus\{0\})$ by continuity.
\end{proof}
Denote by $\check M$ the set of spacelike or timelike vectors: $\check M \equiv M\setminus\Lambda_\pm$.
\begin{corollary}\label{corautom}
Let $\D$ be a (possibly non local) conformal net on $M$. Given $\lambda_L, \lambda_R >0$, there exists an automorphism $\a\equiv \a_{\lambda_L,\lambda_R}$ of $\gD(\check M)$ such that 
\[
\a|_{\gD(W_L)} = 
\delta_{\lambda_L}\ , \quad \a|_{\gD(W_R)} = \delta_{\lambda_R} \ .
\]
Here $\delta_\lambda$ is the $\lambda$-dilation automorphism of $\gD(M)$.
\end{corollary}
\begin{proof}
Similar to the one-dimensional case.
\end{proof}
\begin{proposition}\label{state}
There exists a natural state $\psi \equiv\psi_{\b_L , \b_R}$ on $\gD(\check M)$ such that $\psi|_{\gB^{L/R}(W_{L/R})}$ is $\f^{L/R}_{\b_L/\b_R}$.
\end{proposition}
\begin{proof} 
The state $\psi$ on $\gD(\check M)$ is given by $\psi \equiv \f\cdot\a_{\lambda_L , \lambda_R}$, where
$\f$ is the geometric state on $\gD$ (at inverse temperature $1$) and
 $\a_{\lambda_L , \lambda_R}$ is the automorphism in Corollary \ref{corautom} with 
$\l_L = \b_L^{-1}$, $\l_R = \b_R^{-1}$. 
\end{proof}
It is possible to extend the state $\psi$ in Prop. \ref{state} to a state on $\gD$ by the Hahn-Banach theorem; for simplicity we still denote by $\psi$ this extended state. Cleary $\psi$ is locally normal on $\check M$, but not around the light rays. As we shall see in Sect. \ref{probe}, we can choose a locally normal state $\psi$ on $M$ by inserting a probe.

\medskip

\noindent
{\bf The large time limit.} Waiting a large time we expect the global system to reach a stationary state, a non equilibrium steady state.

The two nets $\B^L$ and $\B^R$ both contain the same net $\A=\A_+\otimes\A_-$. And the chiral net $\A_\pm$ on $\mathbb R$ contains the Virasoro net with central charge $c_\pm$. In particular $\B^L$ and $\B^R$ share the same stress energy tensor. 

Let $\f^+_{\b_L}$, $\f^-_{\b_R}$ be the geometric KMS states respectively on $\gA_+$ and  $\gA_-$ with inverse temperature $\b_L$ and $\b_R$; we define
\[
\omega \equiv\f^+_{\b_L}\otimes\f^-_{\b_R}\cdot\varepsilon\ ,
\] 
so $\o$ is the state on $\gD$ obtained by extending $\f^+_{\b_L}\otimes\f^-_{\b_R}$ from $\gA$ to $\gD$ by the conditional natural expectation $\varepsilon: \gD\to\gA$.
Clearly $\omega$ is a stationary state, indeed:
\begin{proposition}\label{holom}
$\omega$ is a non-equilibrium steady state on $\gD$ with $\b={\rm min}\{\b_L , \b_R\}$.
\end{proposition}
\begin{proof}
The states $\f^+_{\b_L}$ and $\f^-_{\b_R}$ are obtained by 
\[
\f^+_{\b_L} = \f^+\cdot \delta_{\l_{L}}\ ,\quad \f^-_{\b_R} = \f^-\cdot \delta_{\l_{R}}\ ,
\]
where $\f_\pm$ is the geometric KMS on $\gA_\pm$ at inverse temperature 1 and  $ \delta_{\l_{L/R}}$ is the dilation automorphism of $\gA_\pm$ with $\l_{L/R} =\b_{L/R}^{-1}$.

Let $\delta$ be the automorphism $\delta_{\lambda_{L}}\otimes \delta_{\lambda_{R}}$ of $\gA=\gA_+\otimes\gA_-$. Then $\omega|_\gA = \f\cdot \delta$, with $\f=\f^+\otimes\f^-$ the geometric KMS state on $\gA$. Since $\D$ contains the diffeomorphism symmetries, this formula extends to $\gD$, namely
\[
\omega \equiv  \f\cdot \delta\ ,
\]
where we denote by the same symbols also on $\gD$ the geometric KMS state $\f$ and the product dilation automorphism
 $\delta=\delta_{\lambda_{L}}\otimes \delta_{\lambda_{R}}$. Then 
\[
\o\cdot\varepsilon =\o \ ,
\] 
with $\varepsilon$ the conditional expectation onto $\gA$\ .

Now, given $X,Y\in\gD$, we have 
\begin{multline*}
F_{X,Y}(t)  \equiv  \omega\big(X\tau_t(Y)\big)
 =  \f\cdot\delta\big(X\tau_t(Y)\big)
 =  \f\Big(\delta\big(X\tau^+_t\cdot\tau^-_t(Y)\big)\Big) \\
 =  \f\big(X'\tau^+_{\lambda_L t}\cdot\tau^-_{\lambda_R t}(Y')\big) 
  =  \f\big(X'\tau^+_u\cdot\tau^-_v(Y')\big) \ ,
\end{multline*}
with $X'=\delta(X)$,  $Y'=\delta(Y)$ and $u = t/\b_L$, $v=t/\b_R$,
so the holomorphic property of $F_{X,Y}(t)$ follows by the joint holomorphic property in the  variables  $(u,v)$ given by the relativistic KMS property, Theorem \ref{relKMS}.
\end{proof}
We now want to show that the evolution $\psi\cdot\tau_t$ of the initial state $\psi$ of the composite system defined in Lemma \ref{state}  approaches the non-equilibrium steady state $\omega$ as $t\to+\infty$.

Note that:
\begin{lemma}\label{V}
$\psi|_{\D(O)} = \omega|_{\D(O)}$ if  $O\in\K(V_+)$.
\end{lemma}
\begin{proof} The conditional expectation $\varepsilon: \gD\to\A$ leaves $\omega$ invariant, $\omega(\varepsilon(Z)) = \omega(Z)$ for all $Z \in \mathfrak{D}$, because the geometric state is constructed by local diffeomorphism symmetries that intertwine $\varepsilon$. Similarly, $\psi(\varepsilon(Z)) = \psi(Z)$ if $O\in\K(\check M)$.
So we have to show that $\psi|_{\A(O)} = \omega|_{\A(O)}$ if  $O\in\K(V_+)$. Now the double cone $O=I\times J$ is contained in $V_+$ if and only if $\bar I, \bar J \subset \mathbb R^+$, so in this case both states $\psi$ and $\omega$ are equal to $\f^+_{\b_L}\otimes\f^-_{\b_R}$ on $\A_+(I)\otimes\A_-(J)$ (recall that $\A_+(I) \subset \B^R(W_R)$ and $\A_-(J) \subset \B^L(W_L)$ by definition \eqref{movers}).
\end{proof}
Let $O\in\K(M)$ and $Z\in\D(O)$. The time translated double cone  $O_t \equiv O+ (t,0)$  enters and remains in $V_+$ for $t$ larger than a certain time $t_O$. Therefore, 
we immediately get:
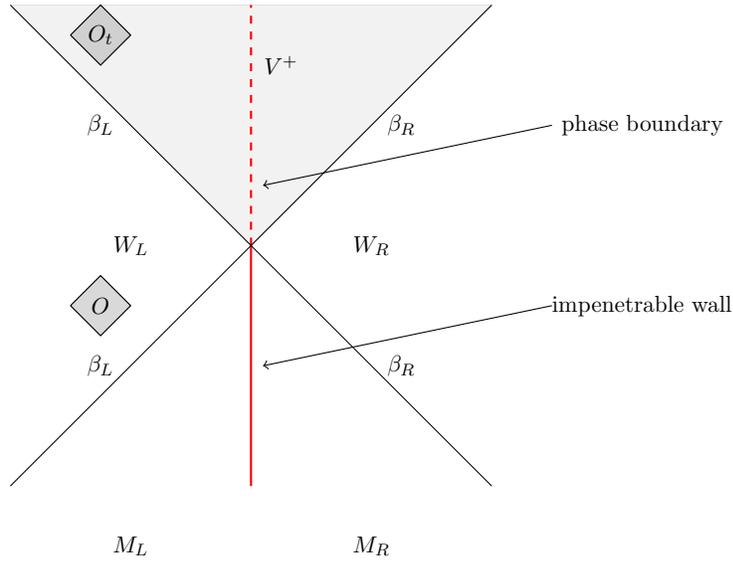
\begin{figure}
\begin{center}
\begin{tikzpicture}[scale=.8, transform shape]
\filldraw[fill=gray,opacity=.1,draw=black] (-4,4) -- (0,0) -- (4,4) -- (-4,4);
\draw[red, thick, dashed] (0,0) -- (0,4);
\draw[red,thick] (0,-4) -- (0,0);
\draw[black] (-4,-4) -- (4,4);
\draw[black] (-4,4) -- (4,-4);
\draw[->,black] (5,2) -- (.2,1);
\draw[->,black] (5,-1) -- (.2,-2);
\filldraw[fill=gray,opacity=.3,draw=black] (-3,-1) -- (-2.5,-.5) -- (-2,-1) -- (-2.5,-1.5) -- (-3,-1);
\filldraw[fill=gray,opacity=.3,draw=black] (-3,3.5) -- (-2.5,4) -- (-2,3.5) -- (-2.5,3) -- (-3,3.5);
\draw[draw=black] (-3,-1) -- (-2.5,-.5) -- (-2,-1) -- (-2.5,-1.5) -- (-3,-1);
\draw[draw=black] (-3,3.5) -- (-2.5,4) -- (-2,3.5) -- (-2.5,3) -- (-3,3.5);
\node at (.5,3) {$V^+$};
\node at (6.5,2) {phase boundary};
\node at (6.5,-1) {impenetrable wall};
\node at (-2.5,-1) {$O$};
\node at (-2.5,3.5) {$O_t$};
\node at (2.5,2) {$\beta_R$};
\node at (-2.5,2) {$\beta_L$};
\node at (2.5,-2) {$\beta_R$};
\node at (-2.5,-2) {$\beta_L$};
\node at (2,-5) {$M_R$};
\node at (-2,-5) {$M_L$};
\node at (2,0) {$W_R$};
\node at (-2,0) {$W_L$};
\end{tikzpicture}
\end{center}
\caption{
\label{DOC2}
Spacetime diagram of simplified setup. There is just one phase boundary and no probe. Every time-translated diamond will eventually 
enter the future lightcone $V^+$. 
}
\end{figure}

\begin{proposition}\label{asynt}
For every $Z\in\gD$ we have:
\[
\lim_{t\to+\infty}\psi\big(\tau_t (Z)\big) =\omega(Z) \ . 
\]
\end{proposition}
\begin{proof}
Let $Z\in\D(O)$ with $O\in\K(M)$. 
If $t>t_O$, we have $\tau_t(Z)\in \gD(V_+)$ as said, so 
\[
\psi\big(\tau_t (Z)\big) =\omega\big(\tau_t(Z)\big) = \omega(Z)\ , \quad t > t_O\ ,
\] 
because of Lemma \ref{V} and the stationarity property of $\omega$. Therefore the limit holds true if $Z$ belongs to the norm dense subalgebra of $\gD$ generated by the $\D(O)$'s, $O\in\K$, hence for all $Z\in\gD$ by the density approximation argument.
\end{proof}
\begin{remark}
The steady state $\o$ is a local thermal equilibrium state in the sense of \cite{BOR} on the Virasoro subnet, as $\o$ is the tensor product of KMS states on $\A$.
\end{remark}

One can readily compute the expectation values in the steady state $\o$ of the flux density operator $T_{tx}(t, x)$ (momentum density component of the stress energy tensor $T_{\mu\nu}$). We can of course write it in terms of the left and right chiral stress energy tensor as
\[
T_{tx}(t,x) = T^+(t-x)-T^-(t+x)\ .
\]
Since the conditional expectation $\varepsilon$ leaves the Virasoro subnet invariant, we thus get for the large-time limit of the flux
\be\label{asyT}
\lim_{t\to+\infty} \psi(T_{tx}(t,x)) =\f^+_{\b_L}(T^+(0)) - \f^-_{\b_R}(T^-(0))
= \frac\pi{12}(c_L \beta_L^{-2} - c_R \beta_R^{-2}) \ ,
\ee
where $c_L$ and $c_R$ are the central charges of the left and right movers. The formula for the expectation value of the chiral stress energy tensor in the geometric KMS state is taken from \cite[Thm. 5.1]{CLTW2}.

\subsection{ Case with chemical potentials and no probe}
In the completely rational case the chemical potential does not appear as the geometric $\b$-KMS state is the unique $\b$-KMS state (Theorem \ref{gstate}).
We indicate here how to generalize Section \ref{main} in the case a chemical potential is present. We recall the following. 
\begin{theorem}{\rm \cite{CLTW2}.}\label{gstate2}
Let $\C$ be the local conformal net on $\mathbb R$ associated with the $U(1)$-current $J$. The translation KMS states $\f_{\b, q}$ on $\gC$ are labeled by the inverse temperature $\b > 0$ and the charge $q\in\mathbb R$. The state $\f_{\b, q}$ is locally normal, $\f_{\b, 0}$ is the geometric KMS state at inverse temperature $\b$ and
\[
\f_{\b, q} = \f_{\b, 0}\cdot \gamma_q\ ,
\]
where $\gamma_q$ is the automorphism of $\C$ corresponding to $J \mapsto J +q$.
\end{theorem}
The automorphism $\gamma_q$ is localizable, i.e. equivalent to a DHR sector \cite{H} (the charge $q$ sector, cf. eq. \eqref{charge} below). Since $\C$ is strongly additive, 
$\f_{\b, q}$ satisfies essential duality and $\gamma_p$ is normal in the GNS representation of $\f_{\b, q}$ \cite{L01}, so the chemical potential is defined as in Sect. \ref{ChemPot}.
\begin{lemma}\label{Pext}
The chemical potential of $\f_{\b,q}$ with respect to the charge $p$ is $qp/\pi$.
\end{lemma}
\begin{proof} 
Recall that the $U(1)$ current algebra may be realized as a (weak closure) of the Weyl algebra generated by the unitaries $W(f), f \in C^\infty_0(\mathbb{R}, \mathbb{R})$, subject 
to the Weyl form of the canonical commutation relations $W(f)W(g) = \exp[i \sigma(f,g)] W(f+g), W(f)^* = W(-f)$, where $\sigma(f,g) =  \pi \int_{\mathbb{R}} fg' \ dx$ is the symplectic form. 
Informally, we may think of $W(f)$ as $\exp[-i\int_{\mathbb{R}} J(x) f(x) dx]$, and the Weyl relations are then formally equivalent to the commutation relations of the 
chiral $U(1)$-current, $[J(x), J(y)] = i\pi \delta'(x-y)1$, see~\cite{BMT} for details. Now
choose a smooth real function with compact support $\ell:\mathbb R\to \mathbb R$ with $\frac1{2\pi}\int_{\mathbb{R}} \ell(s)ds = p$. Then
\be\label{charge}
\gamma_\ell\big((W(f)\big)\equiv e^{-i\int \ell f \ ds}W(f)\ ,
\ee
defines a representative $\gamma_\ell$ of the charge $p$ sector that is localized in any interval containing the support of $\ell$. 

A simple calculation using the Weyl relations shows that we may choose the time-translation covariance unitary cocycle for $\gamma_\ell$ as 
$U(t) = W(L( \ \cdot \ ) - L( \ \cdot \ -t))$, where $L$ is a primitive of $\frac1\pi\ell$. Using again the Weyl relations, we find
\[
\gamma_q\big(U(t)\big) = \exp\left\{-iq\int_{\mathbb{R}}(L(s) - L(s-t)) ds \right\}U(t)\ .
\]
Thus we have
\[
\f_{\b,q}\big(U(t)\big) = \f_{\b,0}\big(\gamma_q\big(U(t)\big) =  \exp\left\{-iq\int_{\mathbb{R}}(L(s) - L(s-t)) ds \right\} \f_{\b,0}\big(U(t)\big)\ ,
\]
so by eq. \eqref{fc}, the chemical potential $\mu_p(\f_{\b,q})$ satisfies
\[
e^{2\pi\b\mu_p(\f_{\b,q})} = {\underset{t\, \longrightarrow\,  i\b}{\rm anal.\, cont.\,}}e^{-iq\int(L(s) - L(s-  t))ds}
= {\underset{t\, \longrightarrow\,  i\b}{\rm anal.\, cont.\,}}e^{-i2qpt} = e^{2\b qp}\ ,
\]
therefore $\mu_p(\f_{\b,q}) = qp/\pi$.
\end{proof}
Since the chemical potential is proportional to $p$, it is determined by its value at $p=1$, so we set
\[
\mu(\f_{\b,q}) \equiv \mu_1(\f_{\b,q}) = q/\pi \ .
\]

We define the state $\o_{\b_1,\b_2,q_1,q_2}$ on $\gC(\mathbb R\setminus\{0\})$ by
\[
\o_{\b_1,\b_2,q_1,q_2} = \f_{\b_1,q_1}|_{\gC(-\infty,0)}\otimes\f_{\b_2,q_2}|_{\gC(0,\infty)}\ .
\]
We now get back in the framework of Section \ref{main}, but we suppose here that $\A_\pm$ is the above net $\C$ generated by the $U(1)$-current $J^\pm$ (thus $\B^{L/R}$ is non rational with central charge $c=1$). 

Notice that, in this case, all irreducible sectors of $\A$ are abelian (i.e. represented by automorphisms), and $\A$ is the fixed-point subnet of $\D$ under a compact abelian gauge group of automorphisms of $\D$ \cite{BMT}. So we may extend KMS states from $\gA$ to $\gD$ \cite{AHKT,CLTW2}.

Given $q\in\mathbb R$, the $\b$-KMS state $\f_{\b,q}$ on $\gD$ with charge $q$ is defined by 
\[
\f_{\b,q} = \f_{\b,q}^+ \otimes\f_{\b,q}^- \cdot\varepsilon\ , 
\]
where $\f_{\b,q}^\pm$ denote the state characterized by the previous lemma and theorem on $\A_\pm$. 
$\f_{\b,q}$ satisfies the $\b$-KMS condition on $\gD$ w.r.t. the one-parameter automorphism group
$t\mapsto\t_t\cdot\a_t$, where $\t$ is the time-translation one-parameter automorphism group of $\D$ and $\a$ a one-parameter subgroup of the gauge group of $\gD$.
\begin{proposition}\label{stat}
Given $\b_{L/R}>0$, $q_{L/R}\in\mathbb R$,
there exists a state $\psi$ on $\gD$ such that 
\[
\psi|_{\gB^{L}(W_{L})}= \f_{{\b_{L},q_{L}}}|_{\gB^{L}(W_{L})}\ ,\qquad \psi|_{\gB^{R}(W_{R})}= \f_{{\b_{R},q_{R}}}|_{\gB^{R}(W_{R})}\ .
\]
\end{proposition}
\begin{proof}
The restriction $\psi_0$ of $\psi$ to $\gA(\check M)=\gC_+(\mathbb R\setminus\{0\})\otimes\gC_-(\mathbb R\setminus\{0\})$ is defined by
\[
\psi_0 \equiv \o_{\b_L,\b_R,q_L,q_R}|_{\gC_+(\mathbb R\setminus\{0\})}\otimes \o_{\b_R,\b_L,q_R,q_L}|_{\gC_-(\mathbb R\setminus\{0\})}\ ;
\]
the restriction of $\psi$ to $\gD(\check M)$ is then defined by $\psi|_{\gD(\check M)} \equiv \psi_0\cdot\varepsilon$.
So any Hahn-Banach extension $\psi$ of $\psi|_{\gD(\check M)}$ to $\gD$ has the properties required by the proposition.
\end{proof}
Now  $\o=\f_{\b_L, q_L}^+\otimes\f_{\b_R, q_R}^-\cdot\varepsilon$ is a steady state (this can be checked by arguments similar to the ones in the proof of Theorem \ref{relKMS}) and $\o$ is evidently determined uniquely by the inverse temperatures $\b_{L/R}$ and the charges $q_{L/R}$
\[
\f^+_{\b_L,q_L}\big(J^+(0)\big) = q_L\ , \qquad \f^-_{\b_R, q_R}\big(J^-(0)\big) = q_R\ .
\]
We also have (see \cite{CLTW2}):
\[
\f_{\b_L,q_L}^+\big(T^+(0)\big) = \frac{\pi}{12\b_L^2} + \frac{q_L^2}2\ , \qquad \f_{\b_R,q_R}^-\big(T^-(0)\big) = \frac{\pi}{12\b_R^2} + \frac{q_R^2}2\ .
\]
In presence of chemical potentials  $ \mu_{L/R} = \frac1{\pi}q_{L/R}$, the large time limit of the two dimensional current density expectation value ($x$-component of 
the current operator $J^\mu$) in the state $\psi$ is, 
with $J^x(t,x) = J^-(t+x) - J^+(t-x)$
\[
\lim_{t\to+\infty} \psi\big(J^x(t,x)\big) =\f_{\b_L,q_L}^-\big(J^-(0)\big) - \f_{\b_R,q_R}^+\big(J^+(0)\big) = -\pi(\mu_L - \mu_R) \ ,
\]
whereas formula \eqref{asyT} becomes here
\[
\lim_{t\to+\infty} \psi\big(T_{tx}(t,x)\big) =\f^+_{\b_L,q_L}\big(T^+(0)\big) - \f^-_{\b_R,q_R}\big(T^-(0)\big)
= \frac\pi{12}\big( \beta_L^{-2} -  \beta_R^{-2}\big) + \frac{\pi^2}2\big(\mu^2_L - \mu^2_R\big) \ ,
\]
which corresponds to formulas obtained in \cite{BD12,BD15} (see also \cite{BD16}).

The above discussion could be extended to the case $\A_\pm$ contains a higher rank current algebra net. Note that, in this case, the net would typically be completely rational.

\subsection{Inserting a probe}\label{probe}

We discuss here how to generalize the discussion in Section \ref{main} when a probe is put between the thermal reservoirs. We do not make a completely rationality assumption here, but we deal with the geometric KMS states (the unique KMS state in the completely rational case). 

With ${\bf a} = (t_1,a)$, ${\bf b} = (t_2, b)$ two points of $M$ with $\bf b$ in the right spacelike complement of $\bf a$ 
(in particular $a < b$), we consider the case of two defect lines, the time axes through the points $\bf a$ and $\bf b$, as described in the introduction~\footnote{
In the introduction, we assumed for simplicity that $t_1=t_2=0$, but the more general situation is more natural and does not cause any additional difficulties.}. 
We assume that we have three local conformal nets $\B^L$, $\B^I$ and $\B^R$ on the same Hilbert space, all containing the same chiral subnet $\A = \A_+\otimes\A_-$. $\B^L$ lives on the left half-plane $M^{a}_L \equiv M_L^{} + {\bf a}$, 
$\B^R$ lives on the left right-plane $M^{b}_R \equiv M_R^{} + {\bf b}$, and $\B^I$ on the intermediate time strip $M^{ab}_I \equiv \{(t,x): a < x < b\}$, yet they can be continued to nets on the entire $M$ as in the previous sections. $\B^L$ and $\B^R$ represents the two reservoirs,  $\B^I$  the probe.

By causality, if the three double cones $O_L , O_I , O_R$ are contained respectively in $M^a_L, M^{ab}_I , M^b_R$ and spacelike separated, then $\B^L(O_L)$, $\B^I(O_I)$, $\B^R(O_R)$ are mutually commuting von Neumann algebras.
\begin{figure}
\begin{center}
\begin{tikzpicture}[scale=.8, transform shape]
\filldraw[fill=gray,opacity=.1,draw=black] (-4.5,9.5) -- (-1,6) -- (2.5,9.5) -- (-4.5,9.5);
\filldraw[fill=blue,opacity=.2,draw=black] (-7,0) -- (5,0) -- (5,6) -- (2,3) -- (-1,6) -- (-3,4)--(-7,8)--(-7,0);
\draw[draw=black] (-1,6)--(-4.5,9.5);
\draw[draw=black] (-1,6)--(2.5,9.5);
\draw[dashed, draw=black] (-7,8) -- (-3,4) -- (-1,6) -- (2,3) -- (5,6);
\draw[dashed, draw=black] (-7,0) -- (-3,4) -- (0,1) -- (2,3) -- (5,0);
\draw[red, thick, dashed] (2,3) -- (2,9.5);
\draw[red, thick, dashed] (-3,4) -- (-3,9.5);
\draw[red,thick] (-3,0) -- (-3,4);
\draw[red,thick] (2,0) -- (2,3);
\draw[->,black, thick] (0,-1) -- (0,9.8);
\draw[->,black, thick] (-8,3.5) -- (7,3.5);
\draw[->,black] (3,8) -- (-2.7,6);
\draw[->,black] (3,7) -- (2.3,6);
\draw[->,black] (3,3) -- (-2.7,1);
\draw[->,black] (3,2) -- (2.3,1);
\node at (-2.5,7.8) {$\beta_L$};
\node at (.5,7.8) {$\beta_R$};
\node at (-.8,7) {$V^+$};
\node at (4.5,7.5) {phase boundaries};
\node at (4.5,2.5) {impenetrable walls};
\node at (6.5,3.8) {$x$};
\node at (.3,9.5) {$t$};
\node at (2.3,3.0) {$\bf b$};
\node at (-3.3,4.0) {$\bf a$};
\node at (4,4) {$W_R^{b}$};
\node at (-6,5) {$W_L^{a}$};
\node at (-1,4.5) {$W_L^{b} \cap W_R^{a}$};
\node at (-1,-1) {$M_I^{ab}$};
\node at (4,-1) {$M_R^{a}$};
\node at (-6,-1) {$M_L^{b}$};
\end{tikzpicture}
\end{center}
\caption{
\label{DOC3}
Spacetime diagram of our setup. The initial state $\psi$ is set up in the shaded region before the system is in causal contact with the phase boundaries. 
In $W^a_L$ resp. $W^b_R$, we have a thermal equilibrium state at inverse temperatures $\beta_L$ resp. $\beta_R$. In the diamond shaped 
shaded region $W_L^b \cap W^a_R$, we have an essentially 
arbitrary probe state. Again, every time-translated causal diamond in this region will eventually enter $V^+$.
}
\end{figure}
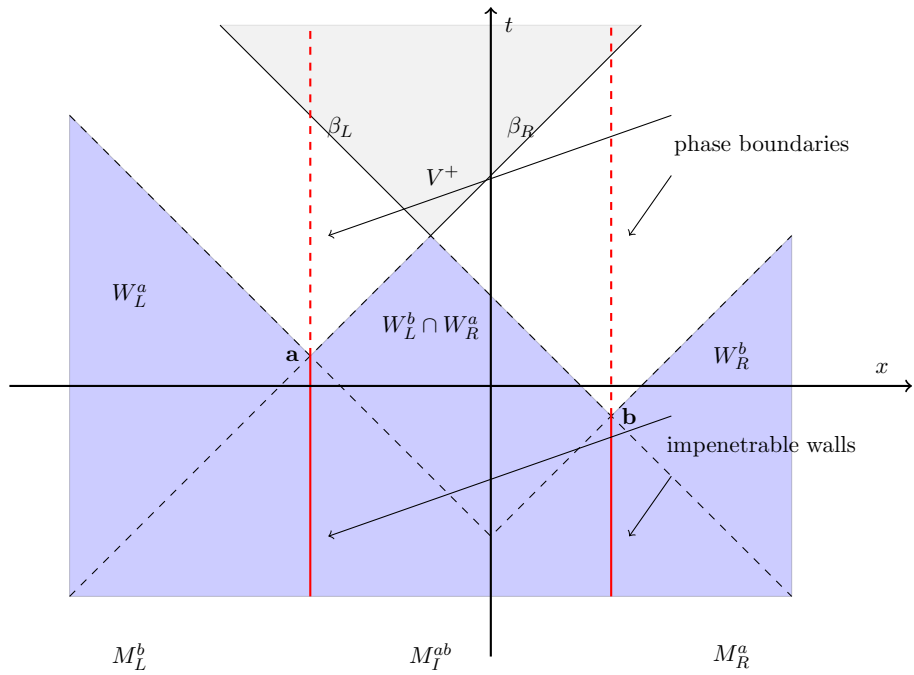
\begin{lemma}
$\B^L$ is left local w.r.t. $\B^I$ and $\B^R$; $\B^I$ is left local w.r.t. $\B^R$.
\end{lemma}
\begin{proof}
Let us show for example that $\B^L$ is left local w.r.t. $\B^I$. It suffices to show that $\B^L(O_1)$ and $\B^I(O_2)$ commute if $O_1\in\K(M^a_L)$ and $O_2\in\K(M^a_R)$ are spacelike separated. We can choose a diffeomorphism symmetry of $M$ acting trivially on $O_1$
and mapping $O_0$ onto $O_2$, with $O_0\in\K(M^{ab}_I)$. As $\B^L(O_1)$ commutes with $\B^I(O_0)$ and with the covariance unitary $U(g)$ implementing $g$, it also commutes with $U(g)\B^I(O_0)U(g)^* = \B^I(O_2)$.
\end{proof}
We now set
\[
\D(O) = \B^L(O)\vee\B^I(O)\vee\B^R(O) \ .
\]
Let $\f_{\b_{L/R}}^{L/R}$ be the geometric $\b_{L/R}$-KMS state of $\gB^{L/R}$. We shall construct a locally normal state $\psi$ on $\gD$ such that
\be\label{pKMS}
\psi|_{\gB^L(W^a_L)} = \f_{\b_L}^L|_{\gB^L(W^a_L)}, \quad \psi|_{\gB^R(W^b_R)} = \f_{\b_R}^R|_{\gB^R(W^b_R)}, \quad \psi|_{\gB^I(W^a_R \cap W^b_L)} \
{\rm unknown},
\ee
where $W^a_L \equiv W_L^{} + {\bf a}$, $W^b_R \equiv W_R^{} + {\bf b}$ and $W^a_R \cap W^b_L$ is the double cone with vertices $\bf a$ and $\bf b$.

We require
\[
 \f_{\b_L}^L|_\gA = \f^+_{\b_L}\otimes\f^-_{\b_L} , \quad  \f_{\b_R}^R|_\gA = \f^+_{\b_R}\otimes\ \f^-_{\b_R}
\]
with $\f^\pm_{\b_{L/R}}$ the geometric $\b_{L/R}$-KMS state of $\gA^\pm$.
\begin{proposition}
There exists a natural, locally normal state $\psi$ on $\gD$ such that \eqref{pKMS} holds.
\end{proposition}
\begin{proof}
With $O = I_+\times I_-$, $I_\pm\subset \L_\pm$, let $h_\pm:\mathbb R\to\mathbb R$ be smooth, strictly increasing maps such that 
$h_+(s) = \l_{L/R} s$ if $s$ is left/right to $I_+$ and $h_-(s) = \l_{R/L} s$ if $s$ is left/right to $I_-$, where $\l_{L/R}\equiv \b_{L/R}^{-1}$.
Then $h_+\times h_-$ implements a locally normal automorphism $\a_{\l_L,\l_R}$ of $\gD$. With $\f$ the geometric KMS state of $\gD$, the state $\psi$ is given by $\psi = \f\cdot \a_{\l_L,\l_R}$.
\end{proof}
Also in this case we have the large time approach to a steady state:
\begin{proposition}\label{asynt2}
$\lim_{t\to+\infty}\psi\cdot\tau_t =\omega$ (weak$^*$-limit in the dual of $\gD$),
where $\o \equiv \f^+_{\b_R}\otimes\o^-_{\b_L}\cdot\varepsilon$.
\end{proposition}
\begin{proof}
Let $O\in\K$ and $Z\in\D(O)$. For sufficiently large $t$, the double cone $O + (t,0)$ is contained in the cone $V_+ +{\bf x}$, with $\bf x$ the upper vertex of $W^a_R \cap W^b_L$. Thus $\psi(\tau_t(Z))= \o(\tau_t(Z)) = \o(Z)$. As the union of all the $\D(O)$ is norm dense in $\gD$, the result follows by the usual approximation argument.
\end{proof}
The case with a probe and chemical potentials is easily described by combining the above discussions.

\section{The relativistic KMS condition in CFT}

Let $\D$ be a conformal net of factors on the Minkowski plane $M$. Following \cite{BB}, we shall say that a state $\f$ on $\gD$ satisfies the relativistic KMS condition at inverse temperature $\b$ w.r.t. the spacetime translation group $\t$ if there is a timelike unit vector $\bf e$ such that for any given $X,Y\in\gD$ there exists a function  $F_{XY}$ bounded analytic in the tube $\I_\b\equiv\{z\in\mathbb C^2: {\rm Im}\, z \in V_+\cap (\beta{\bf e} + V_-)\}$, continuous on the closure of $\I_\b$, such that
\begin{align}\label{RKMS}
F_{XY}({\bf x}) = \f&\big(X\t_{\bf x}(Y)\big)\ , \\ 
F_{XY}({\bf x} + i\b{\bf e}) &= \f\big(\t_{\bf x}(Y)X)\big)\ , \label{RKMS2} 
\end{align} 
with ${\bf x} = (t,x)$. In the following ${\bf e} = (1,0)$, thus ${\bf e} = (1,1)$ in the chiral $u,v$ variables, and the KMS condition reads
\[
F_{XY}(u,v) = \f\big(X\t_{(u,v)}(Y)\big),\quad F_{XY}(u + i\b, v+ i\b) = \f\big(\t_{(u,v)}(Y)X)\big)\ .
\]
\begin{theorem}\label{relKMS}
Let $\f$ be the geometric $\b$-KMS state on $\gD$ w.r.t. the time translation automorphism groups. If $\D$ is completely rational, then $\f$ satisfies the relativistic $\b$-KMS condition w.r.t. the spacetime translation group $\tau$. 
\end{theorem}
\begin{proof}
Let $\tau^+_u = \t_{(u,0)}$, $\tau^-_v = \t_{(0,v)}$ be the chiral translation automorphism group. Clearly $\tau$ restricts to 
$\gA=\gA_+\otimes\gA_-$ and we have the tensor product decomposition $\t = \t^+\otimes\t^-$ on $\gA$. Now $\f= \f^+\otimes\f^-$ on $\gA$, with $\f^\pm$ a $\b$-KMS state of $\A_\pm$ w.r.t. $\t^\pm$ (see \cite{CLTW1}), namely $\f|_{\gA}$ is a tensor product of KMS states,
so the relativistic KMS condition on $\gA$ is readily verified.

To verify the relativistic KMS condition of $\f$ on $\gD$ we fix a double cone $O= I_+\times I_-$ of $M$. Consider the inclusion of factors $\A(O)\subset\D(O)$, and a dual canonical endomorphism $\vartheta$ localized in $O$. By assumptions, $\vartheta |_{\A(O)}$ is the direct sum (with multiplicity) of irreducible DHR endomorphisms, namely $\vartheta = \bigoplus_i \r_i = \bigoplus_i\r^+_i\otimes\r^-_i$. 

Each element $X\in \D(O)$ is a finite sum 
\be\label{fs}
X = \sum_i a_i R_i\ ,
\ee
with $a_i\in \A(O)$ \cite{L03}. Here $R_i\in\D(O)$ is an isometry such that $R_i a =\r_i(a)R_i$ for every $a\in\gA$; we can choose $R_0 = 1$. We have $R_iR_j = \sum_k C_{ij}^k R_k$ with $C_{ij}^k\in\A(O)$.
Note that $\varepsilon(X) = a_0$, where $\varepsilon:\gD\to\gA$ is the natural conditional expectation,  so $\f(X) = \f(a_0)$ because $\f\cdot\varepsilon =\f$.
We denote by $U_i$ the translation unitary covariance cocycle in $\A$ for the endomorphism $\r_i$, in particular
\[
{\rm Ad}U_i({\bf x})\cdot\t_{\bf x} \cdot \r_i  = \r_i \cdot \t_{\bf x} \ ;
\]
the choice of the $U_i$ is unique if we assume that $U_i$ is the restriction of the M\"obius covariance unitary cocycle (see \cite{L01}). It follows from this relation together with 
$R_i a =\r_i(a)R_i$ for every $a\in\gA$ that 
\be
\t_{\bf x}(R_i) = U^*_i({\bf x})R_i \ . 
\ee
We check eq. \eqref{RKMS} with 
$X = \sum a_iR_i$, $Y=\sum b_j R_j$ with $a_i, b_j \in \A(O)$,  $O = I_+ \times I_- \in \mathcal{K}$. 
For this, it is convenient to consider the function $F({\bf x}) = \f\big(\t_{{\bf x}}(X)Y\big)$ rather than $\f\big(X\t_{{\bf x}}(Y)\big)$ (thus $F({\bf x})  = F_{XY}(-{\bf x})$). We have:
\be\label{long}
\begin{split}
F({\bf x}) &= \f\big(\t_{{\bf x}}(X)Y\big) \\
&= \sum_{i,j} \f\big(\tau_{{\bf x}}(a_i R_i)b_j R_j\big) \\
&= \sum_{i,j} \f\big(\tau_{{\bf x}}(a_i) \tau_{\bf x}(R_i)b_j R_j\big) \\
&= \sum_{i,j} \f\big(\tau_{{\bf x}}(a_i) U^*_i({\bf x})R_i b_j R_j\big) \\
&= \sum_{i,j}  \f\big(\tau_{{\bf x}}(a_i) U^*_i({\bf x})\r_i(b_j)R_i R_j\big) \\
& =\sum_{i,j,k} \f\big(\tau_{\bf x}(a_i) U^*_i({\bf x})\r_i(b_j)C^k_{ij}R_k\big) \\
&= \sum_{i,j} \f\big(\tau_{\bf x}(a_i) U^*_i({\bf x})\r_i(b_j)C^0_{ij}\big) \ .
\end{split}
\ee
Now we have 
\[
U_i({\bf x}) = U_i(u,v) = U^+_i(u)U^-_i(v)\ ,
\]
with $U_i^\pm$ the $\t^\pm$ unitary covariance cocycle in $\A_\pm$ for $\r^\pm_i$. 
Moreover, in the GNS representation of $\gA_\pm$ given by $\f^\pm$, $\t^\pm$ extends to the (rescaled) modular group of the weak closure von Neumann algebra and $U^\pm_i$ is the Connes Radon-Nikodym  cocycle \cite{Con1}:
\be\label{Cco}
U^\pm_i(\b s) = d(\r_i^\pm)^{-is}\big(D\f_\pm\cdot\Phi^\pm_i : D\f^\pm\big)_{-s} \ ,\quad s\in\mathbb R\ ,
\ee
with $\Phi^\pm_i$ the left inverse of $\r^\pm_i$ and $d$ the statistical dimension \cite{L97}. In particular
\[
s\in\mathbb R \mapsto \f^\pm\big(\tau^\pm_s(w_1^\pm)U^\pm_i(s)^*w_2^\pm\big)
\]
is the boundary value of a bounded analytic function on the strip $S_{\b}$, continuous up to the boundary, for every $w_1^\pm, w_2^\pm\in\gA_\pm$ \cite{Con1}.

Next, we know that $C^0_{ij}$ intertwines the identity and $\r_i\r_j$, thus, for each given $i,j$, we have $C^0_{ij} = \sum_h t^+_h t^-_h$ with $t^\pm_h\in\A_\pm(I_\pm)$. This sum is finite. Indeed, as we are dealing with finite index endomorphisms, the intertwiner space Hom$({\rm id}, \r_i\r_j)$ is finite dimensional, and one can check that Hom$({\rm id}, \r_i\r_j)$ is the tensor product of Hom$({\rm id}, \r^+_i\r^+_j)$  and Hom$({\rm id}, \r^-_i\r^-_j)$. We can assume, furthermore, without loss of generality, that $a_i=a_i^+a_i^-$, $b_j= b_j^+ b_j^-$, where $a^\pm_i , b_j^\pm \in \A_\pm(I_\pm)$. (Note that $
t^\pm_h$ depend on $i,j$, but we suppress this to lighten our notation.)
Therefore \eqref{long} gives:
\begin{align*}
F(u,v) &= \f\big(\t_{(u,v)}(X)Y\big) \\
&= \sum_{i,j,h} \f\big(\tau^+_u(a_i^+)\tau^-_v(a_i^-) U^+_i(u)^*U^-_i(v)^*\r^+_i(b^+_j)\r^-_i(b_j^+) t^+_h t^-_h\big) \\
&= \sum_{i,j,h} \f^+\big(\tau^+_u(a_i^+)U^+_i(u)^*\r^+_i(b_j^+)t^+_h\big)  \  \f^-\big(\tau^-_v(a_i^-) U^-_i(v)^*\r^-_i(b^-_j) t^-_h\big)\ ,
\end{align*}
where all sums are finite.
Because of the above mentioned holomorphic properties of the Connes Radon-Nikodym cocycles,  $F(u,v)$ is thus separately holomorphic, hence joint holomorphic by Hartog's theorem, in the product of the strips $S_\b\times S_\b$,
namely $F(-{\bf x})$ extends to a function holomorphic in the tube $\I_\b$, continuous on the closure of $S_\b\times S_\b$. This gives the desired holomorphic property of $\f$ for relativistic KMS condition. Property \eqref{RKMS2} then holds because $\f$ is a KMS state  w.r.t. the time translation one-parameter group.
\end{proof}
\noindent
{\bf Remark.} 
Let $\D\supset\A=\A_+\otimes\A_-$ be a conformal net on $M$, with $\D$ relatively local w.r.t. $\A$.
Under general assumptions, the proof in Theorem \ref{relKMS} of the relativistic $\b$-KMS condition still goes through even without the completely rational assumption, for every extremal $\b$-KMS state of $\gD$ w.r.t. time translation that is invariant under the conditional expectation. The assumptions are the following. The dual canonical endomorphism $\vartheta$ of $\A$ is a direct sum of countably many irreducible sectors $\r_i$ with finite index, the finite sums \eqref{fs} form a dense subalgebra of $\D(O)$
and each $\r_i$ is the tensor product  $\r^+_i\otimes\r^-_i$, with $\r_i^\pm$ a sector of $\A_\pm$ (this tensor product decomposition follows by the split property \cite{KLM}).

\vspace{0.5cm}
\noindent
{\bf Aknowledgements.} S.H. acknowledges discussions with T.-P. Hack and R. Verch.
R.L. would like to thank the Universities of G\"ottingen, Leipzig and Tokyo for the hospitality extended to him during January/March 2016, and Y. Tanimoto for comments.

\end{document}